\documentclass{article}

\input glyphtounicode
\usepackage[T1]{fontenc}
\usepackage[utf8]{inputenc}

\usepackage{ifdraft}

\usepackage[english]{babel}

\usepackage{mathtools}
\usepackage{enumitem}
\usepackage[amsmath,thmmarks,hyperref]{ntheorem}

\usepackage{xcolor}
\definecolor{hypercolor}{rgb}{0,0.2,0.7}
\usepackage[breaklinks,final]{hyperref}

\usepackage[numbers,sort&compress]{natbib}
\newif\iffancyfont%
\fancyfontfalse%

\IfFileExists{MinionPro.sty}{\IfFileExists{MyriadPro.sty}{\fancyfonttrue}{}}{}

\iffancyfont%
  \usepackage[lf,medfamily,italicgreek,openg]{MinionPro}
  \usepackage[lf,medfamily,onlytext,scale=0.96]{MyriadPro}

  \DeclareMathAlphabet{\mathbfit}{OML}{\rmdefault}{bx}{it}
  \DeclareMathAlphabet{\mathsfit}{OML}{\sfdefault}{m}{it}
  \SetMathAlphabet{\mathsfit}{bold}{OML}{\sfdefault}{bx}{it}
  \DeclareMathAlphabet{\mathsfbfit}{OML}{\sfdefault}{bx}{it}

  \DeclareRobustCommand{\defn}{\mathrel{{\vdotdot}{\equal}}}

  \usepackage[toc,eqno,bib]{tabfigures}
\else%
  \usepackage[charter,greekuppercase=italicized,cal=cmcal]{mathdesign}
  \renewcommand{\sfdefault}{phv}

  \newcommand{\defn}{\coloneqq}

  \let\uppi\piup%
  \let\upDelta\Delta%
\fi%

\usepackage{dsfont}

\usepackage[final]{microtype}

\usepackage[sf,bf,small]{titlesec}
\usepackage[labelfont={sf,bf},labelsep=period]{caption}

\setlist[1]{labelindent=\parindent}
\setlist[description]{font=\sffamily\bfseries,align=right,labelsep=1em}

\numberwithin{equation}{section}

\makeatletter
\newcounter{and}
\newdimen{\instindent}
\newcommand{\institute}[1]{\newcommand{\@institute}{#1}}
\newcommand{\inst}[1]{\unskip\smash{$^{#1}$}\setcounter{and}{1}\ignorespaces}
\newcommand{\email}[1]{\href{mailto:#1}{#1}}
\renewcommand{\maketitle}{
  {
    \raggedright%
    \LARGE%
    \noindent%
    \bfseries%
    \sffamily%
    \@title%
    \par
  }

  \vspace{1.5\baselineskip}

  {
    \raggedright%
    \renewcommand{\and}{\unskip, \ignorespaces}%
    \noindent\ignorespaces\@author\par
  }

  \vspace{0.5\baselineskip}

  {
    \small%
    \parindent=0pt%
    \parskip=0pt%
    \setcounter{and}{1}%
    \renewcommand{\and}{%
      \par\stepcounter{and}%
      \hangindent\instindent%
      \noindent%
      \hbox to \instindent{\hss\smash{$^{\theand}$\enspace}}\ignorespaces%
    }%
    \setbox0=\vbox{\@institute}%
    \ifnum\value{and}>9\relax\setbox0=\hbox{$^{88}$\enspace}%
    \else\setbox0=\hbox{$^{8}$\enspace}\fi%
    \instindent=\wd0\relax%
    \ifnum\value{and}=1\relax%
    \else%
      \setcounter{and}{1}%
      \hangindent\instindent%
      \noindent%
      \hbox to \instindent{\hss\smash{$^{\theand}$}\enspace}\ignorespaces%
    \fi%
    \ignorespaces%
    \@institute\par
  }
}
\makeatother

\renewenvironment{abstract}{
  \addvspace{1.5\baselineskip}%
  \topsep=0pt\partopsep=0pt%
  \trivlist\item[\hspace{\labelsep}\bfseries\sffamily Abstract.]
}{}

\newenvironment{acknowledgments}{
  \addvspace{1.5\baselineskip}%
  \topsep=0pt\partopsep=0pt%
  \trivlist\item[\hspace{\labelsep}\bfseries\sffamily Acknowledgments.]
}{}

\usepackage[obeyDraft]{todonotes}

\ifdraft{
  \usepackage[notref,notcite]{showkeys}

  \usepackage{filemod}
  \usepackage{fancyhdr}
  \pagestyle{fancy}
  \fancyhf{}
  
  \fancyhead[C]{Draft version. Last modified: \filemodprint{\jobname}}
  \fancyfoot[C]{\thepage}
}{}

\newcommand{\ie}{\textit{i.e.}}
\newcommand{\eg}{\textit{e.g.}}
\newcommand{\viz}{\textit{viz.}}

\newcommand{\dif}{\mathrm{d}}
\newcommand{\e}{\mathrm{e}}
\newcommand{\im}{\mathrm{i}}

\newcommand{\field}[1][K]{{\mathds{#1}}}

\newcommand{\RR}{{\field[R]}}
\newcommand{\CC}{{\field[C]}}

\renewcommand{\Re}{\operatorname{Re}}
\renewcommand{\Im}{\operatorname{Im}}

\DeclareMathOperator{\spec}{sp}
\DeclareMathOperator{\reso}{rs}

\DeclareMathOperator{\Dom}{Dom}
\DeclareMathOperator{\Ran}{Ran}
\DeclareMathOperator{\Ker}{Ker}

\newcommand{\lapl}{\upDelta}

\newcommand{\comp}{\mathrm{c}}

\newcommand{\slim}{\operatornamewithlimits{s-lim}}

\newcommand{\conj}[1]{\overline{#1}}
\DeclarePairedDelimiter{\abs}{\lvert}{\rvert}
\DeclarePairedDelimiter{\norm}{\lVert}{\rVert}
\DeclarePairedDelimiter{\jnorm}{\langle}{\rangle}

\newcommand{\one}{\mathds{1}}
\newcommand{\cl}{\mathrm{cl}}

\newcommand{\Feyn}{\mathrm{F}}
\newcommand{\aFeyn}{{\overline{\mathrm{F}}}}
\newcommand{\PJ}{\mathrm{PJ}}
\newcommand{\en}{\mathrm{en}}
\newcommand{\retadv}{{\vee\mkern-2mu/\mkern-2mu\wedge}}
\newcommand{\FeynFeyn}{{\Feyn/\mkern1mu\aFeyn}}
\newcommand{\KG}{K}

\makeatletter
\newcommand{\aotimes}{\mathbin{\mathop{\otimes}\limits^{
  \vbox to .15ex {\kern-2\ex@\hbox{$\scriptscriptstyle\mathrm{alg}$}\vss}}}}
\makeatother

\makeatletter
\newtheoremstyle{nonumberplainnoparens}%
  {\item[\theorem@headerfont\hskip\labelsep ##1\theorem@separator]}%
  {\item[\theorem@headerfont\hskip\labelsep ##1 ##3\theorem@separator]}
\makeatother

\theoremnumbering{arabic}
\qedsymbol{\ensuremath{\quad\Box}}

\theoremstyle{plain}

\theorembodyfont{\itshape}
\theoremheaderfont{\normalfont\sffamily\bfseries}
\theoremseparator{.}

\newtheorem{theorem}{Theorem}[section]
\newtheorem{proposition}[theorem]{Proposition}
\newtheorem{lemma}[theorem]{Lemma}

\theorembodyfont{\normalfont}

\newtheorem{remark}[theorem]{Remark}
\newtheorem{assumption}[theorem]{Assumption}

\theoremstyle{nonumberplainnoparens}
\theoremsymbol{\ensuremath{\quad\Box}}
\newtheorem{proof}{Proof}

\let\qedhere\null

\title{Feynman Propagators on Static Spacetimes}

\author{
  Jan Dereziński\inst{1}
  \and
  Daniel Siemssen\inst{1,2}
}

\institute{
  Department of Mathematical Methods in Physics, Faculty of Physics, University of Warsaw, Pasteura 5, 02-093 Warszawa, Poland.
  E-mail:~\email{jan.derezinski@fuw.edu.pl}.
  \and
  Department of Mathematics and Informatics, University of Wuppertal, Gaußstraße 20, 42119 Wuppertal, Germany.
  E-mail:~\email{siemssen@uni-wuppertal.de}.
}

\hypersetup{
  unicode,
  colorlinks,
  linkcolor=hypercolor,
  urlcolor=hypercolor,
  citecolor=hypercolor,
  bookmarksnumbered,
  bookmarksdepth=2,
  pdfauthor={Jan Dereziński, Daniel Siemssen},
  pdftitle={Feynman Propagators on Static Spacetimes}
}

\begin{document}

\maketitle

\begin{abstract}
  We consider the Klein--Gordon equation on a static spacetime and minimally coupled to a static electromagnetic potential.
  We show that it is essentially self-adjoint on $C^\infty_\comp$.
  We discuss various distinguished inverses and bisolutions of the
  Klein--Gordon operator, focusing on the so-called Feynman propagator.
  We show that the Feynman propagator can be considered the boundary value of the resolvent of the Klein--Gordon operator, in the spirit of the limiting absorption principle known from the theory of Schrödinger operators.
  We also show that the Feynman propagator is the limit of the inverse
  of the Wick rotated Klein--Gordon operator.
\end{abstract}

\section{Introduction}
\label{sec:introduction}

Consider a Lorentzian manifold $(M, g)$, an electromagnetic potential $A$ and a scalar potential $Y$.
We write $\abs{g} = \abs{\det [g_{\mu\nu}]}$ and $D = -\im\partial$.
The \emph{Klein--Gordon operator on $(M, g)$ minimally coupled to $A$ and with a scalar potential $Y$} is given by
\begin{equation*}
  \KG = \Box_A + Y = \abs{g}^{-\frac12} (D_\mu - A_\mu) \abs{g}^\frac12 g^{\mu\nu} (D_\nu - A_\nu) + Y
\end{equation*}
and the \emph{Klein--Gordon equation} is
\begin{equation}\label{eq:kg}
  \KG u = 0.
\end{equation}
We are interested in distinguished inverses and bisolutions of the Klein–Gordon operator~$\KG$.
Our main motivation comes from quantum field theory on a fixed curved background and external classical fields.

Inverses and bisolutions of $\KG$ are operators, which often can be interpreted as operators acting from $C^\infty_\comp(M)$ to $C^\infty(M)$, defined by the following conditions:
\begin{enumerate}
  \item We say that $G$ is a \emph{bisolution} of $\KG$ if it satisfies
    \begin{equation*}
      \KG G f = G \KG f = 0
      \quad\text{for all}\quad
      f \in C^\infty_\comp(M).
    \end{equation*}
  \item We say that $G$ is an \emph{inverse} of $\KG$ if it satisfies
    \begin{equation*}
      \KG G f = G \KG f = f
      \quad\text{for all}\quad
      f \in C^\infty_\comp(M).
    \end{equation*}
\end{enumerate}
The Klein--Gordon equation has many bisolutions and inverses.
They have many names, often not quite consistent.
In physics one often uses the word ``propagator'' or ``two-point function''.
Moreover, inverses are often called ``Green's functions''.
We sometimes use the word ``propagator'' to denote jointly distinguished bisolutions and inverses.
An interesting table comparing conventions for propagators used by various authors can be found at the end of Appendix 2 of \cite{bogoliubov}.

In this article we are interested in \emph{distinguished inverses and bisolutions} of the Klein--Gordon operator on certain static spacetimes.
We remark that it is well understood how to define the distinguished bisolutions and inverses in that case.

Here is a list of basic distinguished bisolutions and inverses in the
static case:
\begin{enumerate}
  \item Distinguished bisolutions:
  \begin{enumerate}
    \item the \emph{Pauli--Jordan bisolution}, also called the causal propagator, the commutator function, etc., denoted~$G^\PJ$;
    \item the \emph{positive frequency bisolution}/two-point function, denoted~$G^{(+)}$;
    \item the \emph{negative frequency bisolution}/two-point function, denoted $G^{(-)}$.
  \end{enumerate}
  \item Distinguished inverses:
  \begin{enumerate}
    \item the \emph{forward/retarded inverse}/propagator, denoted~$G^\vee$;
    \item the \emph{backward/advanced inverse}/propagator, denoted~$G^\wedge$;
    \item the \emph{Feynman inverse}/propagator/two-point function, called the causal Green's function in \cite{bogoliubov}, denoted~$G^\Feyn$;
    \item the \emph{anti-Feynman inverse}/propagator/two-point function, denoted~$G^{\aFeyn}$.
  \end{enumerate}
\end{enumerate}

The Pauli--Jordan, forward and backward propagators are best known and they have the most satisfactory theory.
Their application is in the Cauchy problem of the classical theory.
Therefore, we call them \emph{classical propagators}.
In particular, they can be uniquely generalized to the non-static case, under the rather general assumption that the spacetime is globally hyperbolic.

The situation is more complicated for the remaining propagators, which we call \emph{non-classical propagators}.
In contrast to the classical propagators, in a non-static setup non-classical propagators do not have obvious unique definitions.

The main motivation for non-classical propagators comes from quantum field theory.
This is perhaps an additional reason why they have been much less studied in mathematical literature.
One of the exceptions is a paper by Duistermaat--Hörmander~\cite{duistermaat}, which considers inverses of the Klein--Gordon operator (and more generally of differential operators of real principal type) modulo a smoothing operator.
Such approximative inverses are called \emph{parametrices}.
Duistermaat and Hörmander prove that \emph{Feynman parametrices} can be defined in a large generality.

Similarly to the Feynman propagator, the notion of a positive/negative frequency bisolution has been weakened under the name of a \emph{Hadamard state}.
There exists a considerable literature about them.
Concerning their general properties we would like to mention~\cite{radzikowski}, see also~\cite{khavkine-moretti} and references therein.
Hadamard states have been constructed using various methods, see \eg\ \cite{moretti:hadamard,sahlmann-verch:passive,gerard-wrochna:hadamard,brum-fredenhagen}.

It is well known that on a generic (globally hyperbolic) spacetime one can define the algebra of fields $\hat\psi(x),\hat\psi^*(x)$ (we use here the charged formalism, see \eg~\cite{derezinski}).
It is often stressed in the literature that on such spacetimes there is no distinguished Feynman propagator nor a distinguished Hadamard state.

However, it is also well-known (and important) that on static spacetimes there is
a distinguished Feynman propagator $G^\Feyn$ and a distinguished positive frequency bisolution $G^{(+)}$ -- those that we study in our paper.
This $G^{(+)}$ satisfies the Hadamard condition~\cite{fulling:hadamard,sahlmann-verch:passive} and it can be used to define the physically natural (time-translation invariant) vacuum state $\Omega$, so that we have the relations
\begin{align*}
  \bigl(\Omega\,\big|\,\hat\psi^*(x)\hat\psi(y)\Omega\bigr) &= G^{(+)}(x,y),\\
  \bigl(\Omega\,\big|\,\mathrm{T}\bigl(\hat\psi^*(x)\hat\psi(y)\bigr)\Omega\bigr) &= G^\Feyn(x,y).
\end{align*}

In this article we consider only the static case.
It can be viewed as an introduction to the non-static case, where the question about the possibility of defining distinguished non-classical propagators is much more complicated.

There exists large literature about the Klein--Gordon equation on curved spacetimes, see \eg~\cite{kay,bgp,derezinski-siemssen:evolution}.
However, we think that our paper offers some novel conceptual points on this subject.
To our knowledge, our paper is essentially the first in the mathematically rigorous literature that considers the Klein--Gordon operator as an operator on the Hilbert space $L^2(M)$, where the time extends from $-\infty$ to $+\infty$, and asks about its self-adjointness.
(Recall that $M$ denotes the spacetime).

One could say that considering the Klein--Gordon operator as a self-adjoint operator on $L^2(M)$ is an artificial mathematical question.
We show that this is not the case.
Our main result says that the Feynman propagator (of obvious physical importance) coincides with the boundary value of the resolvent (see Thm.~\ref{thm:main-theorem}).

Note that the Klein--Gordon operator is automatically Hermitian (symmetric).
Therefore, its spectrum coincides with the whole complex plane, the upper or lower halfplane, or is a subset of the real line.
The last case is true if and only if the Klein-Gordon operator is self-adjoint.
Thus its resolvent exists above and below the real axis (so that we can consider its boundary values) only if it is self-adjoint.

Our paper is restricted to the static case, which allows for major simplifications.
However, the questions that we pose (the self-adjointness of the Klein--Gordon operator, the existence of the boundary values of the resolvent and its relationship to the Feynman propagator) can be formulated for non-static spacetimes.
Thus, our paper points towards non-trivial further questions, of physical relevance, which we plan to investigate \cite{derezinski-siemssen:evolution,derezinski-siemssen:feynman}.
Note in particular, that the question of the self-adjointness of a non-static  Klein--Gordon operator is much more difficult from the static case.
In particular, our proof breaks down in a non-static situation.

Most of the literature about the Klein--Gordon operators on curved spacetimes does not consider an electrostatic potential and a variable term in front of $\dif t^2$ (called $V$, resp. $\beta$ in our paper).
If $\beta=1$ and $V=0$ most statements of our paper become easy (and can essentially be found in Sect.~18.3.10 of~\cite{derezinski}).
Including non-trivial $\beta$ and $V$ makes some of our proofs considerably more complicated.
In particular, we need to use some elements of the theory of bisectorial operators, see Sect.~\ref{sec:limiting-absorption}.

To our knowledge, in the mathematical literature the Klein--Gordon operator is rarely considered in the setting of $L^2(M)$.
Some of the recent results of Vasy and his collaborators~\cite{gell-redman,vasy} and of G\'{e}rard and Wrochna~\cite{gerard-wrochna:feynman} about Feynman parametrices can be interpreted in this way.

In some mathematical papers the Klein-Gordon operator is considered on spacetimes with time from a bounded open interval.
This is used, in particular, in some papers devoted to Sorkin--Johnston states, see \eg~\cite{brum-fredenhagen,fewster-verch}.
Restricting to a finite time interval introduces a non-physical question about boundary conditions at the begining and the end of time.
From the point of view of questions asked in our paper it is important that we consider time from $-\infty$ to $+\infty$.

The idea of considering the Klein-Gordon operator as a self-adjoint operator on $L^2(M)$ can be found in the physics literature.
The resolvent of the Klein-Gordon operator with constant external electromagnetic fields is an important ingredient of the famous computation of the effective action due to Schwinger, described \eg\ in Sect.~4.3.3 of~\cite{itzykson-zuber}.
An interesting, partly heuristic analysis of the Feynman propagator on a non-static spacetime was done by Rumpf and his collaborators in~\cite{rumpf1,rumpf2}.
In all these works the self-adjointness of the Klein-Gordon operator was taken for granted, even if it was not always obvious.

The self-adjointness of the spatial part of the Klein-Gordon operator, that is of the magnetic Laplace-Beltrami operator, is well understood~\cite{shubin,strichartz,devinatz,frehse,chernoff}.
It belongs to the domain of elliptic operators, which is not the main topic of our paper, therefore we include it in abstract assumptions.
The main novelty and difficulty of the operator considered in our paper is the fact that it comes from a hyperbolic equation, which does not have a fixed sign.
This causes problems which are non-existent for elliptic operators.

In our paper we make rather weak assumptions on the differentiability of the metric and the potentials.
One of the reasons for doing this is our desire to illustrate the advantages of  our approach to the construction of propagators, based on Hilbert space methods.
Of course, this approach is in principle well-known and belongs to the folklore of the subject.
It is used \eg\ in~\cite{derezinski,gerard-wrochna:feynman}.

In the last section we show that the Feynman propagator can be obtained with help of the Wick rotation.
This easy and essentially well-known fact, mentioned \eg\ in the case $\beta=1$, $V=0$ in Sect.~18.3.10 of~\cite{derezinski}, can be viewed as yet another argument why the Feynman propagator is so important and natural.
However, the Wick rotation can be defined only in static situations, whereas the construction of the Feynman propagator through the boundary value of the resolvent may work in more generality.

\subsection*{Notation and conventions}

Throughout this paper we use the following notation and conventions:

Suppose that $T$ is an operator on a Banach space $\mathcal{X}$.
We denote by $\Dom T$ its domain and by $\Ran T$ its range.
If $T$ is closable, its closure is $T^\cl$.
For its spectrum we write $\spec T$ and for the resolvent set $\reso T$.
$\Dom T$ is equipped with the norm $\norm{u}_T \defn \sqrt{\norm{Tu}^2+\norm{u}^2}$.

Now, suppose that $T$ is an operator on a Hilbert space $\mathcal{H}$ with inner product $(\,\cdot\;|\;\cdot\,)$.
If $T$ is positive, \ie, $(u \,|\, T u) \geq 0$, we write $T \geq 0$.
If also $\Ker T = \{0\}$, then we write $T > 0$.

We denote by $\aotimes$ the algebraic tensor product and by $\otimes$ its Hilbert space completion, which we call \emph{the} tensor product.

We say that $T$ is \emph{dissipative} if its numerical range is contained in the lower complex plane, \viz, $\Im (u \,|\, T u) \leq 0$ for $u \in \Dom T$.
If, additionally, $T$ is closed, densely defined and $\Ran (A - z) = \mathcal{H}$ for some $\Im z > 0$, then $T$ is \emph{maximally dissipative}.

The $p$-times continuously differentiable $\mathcal{X}$-valued functions on a manifold $M$ are denoted $C^p(M; \mathcal{X})$; if $\mathcal{X} = \CC$, we simply write $C^p(M)$.
Sets of compactly supported resp. bounded functions are indicated by a subscript `c' resp. `b'.
In the case of vector bundles we use the same notation but consider sections instead, \eg, $C^1(T^* M)$ denotes the continuously differentiable $1$-forms.
$\mathcal{D}'(M)$ denotes the space of distributions on $M$ and $\mathcal{D}_\comp'(M)$ stands for the space of distributions of compact support.

If $M$ is an orientable manifold and $\gamma$ a positive density (or a pseudo-density on a non-orientable manifold), we denote by $L^2(M, \gamma; \mathcal{X})$ the space of square-integrable $\mathcal{X}$-valued functions.
That is, $L^2(M, \gamma; \mathcal{X})$ is the completion of $C^\infty_\comp(M; \mathcal{X})$ with respect to the norm $\int_M \norm{\,\cdot\,}^2 \gamma$.
If $\mathcal{X} = \CC$, we omit it, and, if $\gamma$ is clear from the context, we omit it as well.
Often we consider the Hilbert space $L^2(M, \gamma)$ with the usual scalar product denoted by
\begin{equation*}
  (u \,|\, v) \defn \int_M \conj{u}\, v\, \gamma.
\end{equation*}
We recall that, given a semi-Riemannian metric~$g$ on~$M$, a natural density is given by~$\abs{g}^\frac12$.

Consider a manifold $M$ and let $A \in C^1(T^* M)$.
If $g$ is a Riemannian metric on $M$, we call $\lapl_A$, locally defined by ($D = -\im\partial$)
\begin{equation*}
  \lapl_A \defn \abs{g}^{-\frac12} (D_i - A_i) \abs{g}^\frac12 g^{ij} (D_j - A_j),
\end{equation*}
the \emph{(magnetic) Laplace--Beltrami operator}.
Adding a scalar potential, $\lapl_A + Y$ is a general form of a \emph{(magnetic) Schrödinger operator}.
If $g$ is instead Lorentzian (we adopt the signature convention $\mathord{-}\mathord{+}\dots\mathord{+}$), we locally define
\begin{equation*}
  \Box_A \defn \abs{g}^{-\frac12} (D_\mu - A_\mu) \abs{g}^\frac12 g^{\mu\nu} (D_\nu - A_\nu)
\end{equation*}
and call it the \emph{(electromagnetic) d'Alembertian}.
Adding a scalar potential $Y$ to the d'\-\nobreak\hspace{0pt}Alembertian, the \emph{(electromagnetic) Klein--Gordon operator} is $\KG \defn \Box_A + Y$.

\section{Klein--Gordon operator on a static spacetime}

Henceforth we shall assume
\begin{assumption}\label{asm:spacetime}
  $(M = \RR \times \Sigma, g)$ is a \emph{standard static spacetime}, \viz, its metric can \emph{globally} be written in the form
  \begin{equation}\label{eq:metric}
    g = -\beta\, \dif t^2 + g_\Sigma,
  \end{equation}
  where $\beta \in C^2(\Sigma)$ is positive and $g_\Sigma$ restricts to a (time-independent) Riemannian metric of class $C^2$ on $\Sigma$.
  Additionally we require that there exists $C > 0$ such that $C \leq \beta \leq C^{-1}$.
\end{assumption}

We consider the Klein--Gordon equation on $(M, g)$ minimally coupled to a static electromagnetic potential $A$ and with a static scalar potential $Y$.
To avoid unnecessarily baroque notation, we write $L^2(\Sigma) = L^2(\Sigma, \beta^\frac12 \abs{g_\Sigma}^\frac12)$ and $L^2(M) = L^2(M, \abs{g}^\frac12)$.
We assume the following properties for $A$ and $Y$:
\begin{assumption}\label{asm:A_and_m}
  $A \in C^1(T^* M)$ with $V \defn -A_0$ bounded, and $Y \in L^2_{\mathrm{loc}}(M)$ positive.
  $A$ and $Y$ are static, \viz, they do not depend on time.
\end{assumption}

Under these assumptions, we have locally (\viz, in a local coordinate chart)
\begin{equation*}
  \KG = -\frac{1}{\beta} (D_t + V)^2 + \abs{g}^{-\frac12} (D_i - A_i) \abs{g}^\frac12 g^{ij} (D_j - A_j) + Y.
\end{equation*}

The factor $\beta^{-1}$ in front of the time derivatives turns out to be a nuisance.
Therefore, instead of working directly with $\KG$, it is often more convenient to consider the operator
\begin{equation*}
  \tilde\KG \defn \beta^\frac12 \KG \beta^\frac12 = -(D_t + V)^2 + L,
\end{equation*}
where
\begin{align*}
  L &\defn \beta^\frac12 \abs{g}^{-\frac12} (D_i - A_i) \abs{g}^\frac12 g^{ij} (D_j - A_j) \beta^\frac12 + \tilde{Y}, \\
  \tilde{Y} &\defn \beta Y.
\end{align*}
Clearly the equation
\begin{equation}\label{eq:kg2}
  \tilde\KG u = 0
\end{equation}
is equivalent to~\eqref{eq:kg}: if $u$ solves~\eqref{eq:kg2}, then $\beta^{-\frac12} u$ solves~\eqref{eq:kg}.

We understand both $\KG$ and $\tilde\KG$ as operators on $L^2(M)$ with domain $C^2_\comp(M)$.
Since $C \leq \beta \leq C^{-1}$, we have that $\KG$ and $\tilde\KG$ share many properties.
In particular, $\tilde\KG$ is Hermitian and if $\KG$ is essentially self-adjoint on $C^2_\comp(M)$ then, by Lem.~\ref{lem:rel_selfadj}, $\tilde\KG$ is essentially self-adjoint on $C^2_\comp(M)$, too.
Note, however, the subtlety that generally $\Dom \KG^* \neq \Dom \tilde\KG^* = \beta^{-\frac12} \Dom \KG^*$.

One of our main assumptions for the remainder of this article is that

\begin{assumption}\label{asm:L_selfadj}
  $L$ is essentially self-adjoint on $C^\infty_\comp(\Sigma)$ with respect to $L^2(\Sigma)$.
  We do not distiguish in notation between $L$ and its closure.
\end{assumption}

\begin{remark}
  If $(\Sigma, g_\Sigma)$ is a complete Riemannian manifold, we see no obvious obstruction to showing the essential self-adjointness of the Schrödinger operator
  \begin{equation*}
    -\lapl_{\vec{A}} + Y = \abs{g}^{-\frac12} (D_i - A_i) \abs{g}^\frac12 g^{ij} (D_j - A_j) + Y
  \end{equation*}
  on $C^\infty_\comp(\Sigma)$, even if the metric and the volume form are only $C^2$.
  We were however unable to find a reference that discusses the self-adjointness in such a low regularity situation.
  In the case where $g_\Sigma$ and $\beta$ are smooth, this follows from~\cite{shubin}.
  For $Y = 0$, $\vec{A} = 0$, $\beta = 1$ and with a $C^2$ metric $g_\Sigma$, this follows from~\cite{strichartz}.
\end{remark}

\begin{remark}
  Suppose $M = \RR^{n+1}$ and choose global Cartesian coordinates.
  Then, under relatively general assumptions (\eg, $Y$ in $L^2_{\mathrm{loc}}$ and bounded below, $\vec{A}$ in $C^1$, $g_\Sigma$ is locally $C^{1,\alpha}$ [Hölder continuously differentiable] and in every open ball there exists $K > 0$ constant such that $K g_\Sigma$ is bounded from below by the Euclidean metric), the Schrödinger operator $-\lapl_{\vec{A}} + Y$ is essentially self-adjoint on $C^\infty_\comp(\Sigma)$, see in particular~\cite{devinatz,frehse}.
\end{remark}

Given our assumption~\ref{asm:L_selfadj}, it is not difficult to show the self-adjointness of $\KG$ using Nelson's commutator theorem:

\begin{theorem}\label{thm:kg_selfadj}
  The Klein--Gordon operator $\KG$ is essentially self-adjoint on~$C^2_\comp(M)$ with respect to~$L^2(M)$.
\end{theorem}
\begin{proof}
  By Lem.~\ref{lem:rel_selfadj}, it is equivalent to show that $\tilde\KG$ is essentially self-adjoint on $C^2_\comp(M)$.
  We apply Nelson's commutator theorem (Thm.~\ref{thm:nelson}) with the Hermitian auxiliary operator
  \begin{equation*}
    N \defn (D_t - V)^2 + L - 2 V^2
  \end{equation*}
  on the dense subspace $\mathcal{C} \defn C^\infty_\comp(\RR; C^2_\comp(\Sigma)) \subset L^2(M)$.
  For this we check essential self-adjointness of $N$ on $\mathcal{C}$ and the conditions (i), (ii) of the theorem.

  Write $L^2(M) = L^2(\RR) \otimes L^2(\Sigma)$ and define the Hermitian operator $N_0 = D_t^2 \otimes \one + \one \otimes L$ on $\mathcal{C}$.
  We can then apply Thm.~\ref{thm:tens_prod_op} to see that $N_0$ is essentially self-adjoint on $C^\infty_\comp(\RR) \aotimes C^2_\comp(\Sigma)$.
  Clearly, $\mathcal{C} \subset \Dom N_0^*$ so $N_0$ is even essentially self-adjoint on $\mathcal{C}$.

  Let $u \in \mathcal{C}$ be arbitrary.
  Since $L \geq 0$,
  \begin{equation*}
    \norm{D_t u}^2 = (u \,|\, D_t^2 u) \leq (u \,|\, N_0 u) \leq \norm{u} \norm{N_0 u}
  \end{equation*}
  and thus for any $\varepsilon > 0$
  \begin{equation}\label{eq:partial_t_est}
    \norm{D_t u} \leq \varepsilon \norm{N_0 u} + \frac{1}{2 \varepsilon} \norm{u}.
  \end{equation}
  In particular this holds for $\varepsilon < 1$, \ie, $D_t$ has relative $N_0$-bound smaller than $1$.
  We can now deduce from the boundedness of $V$ that $N = N_0 - 2 V D_t - V^2$ is also essentially self-adjoint on $\mathcal{C}$.

  (i):
  It follows from the same estimate~\eqref{eq:partial_t_est}, that condition (i) is equivalent to
  \begin{equation*}
    \norm{\tilde\KG u} \leq a \norm{N_0 u} + b \norm{u}.
  \end{equation*}
  We have
  \begin{equation*}
    \norm{(-D_t^2 + L) u}^2 = \norm{(D_t^2 + L) u}^2 - 4 (D_t u \,|\, L D_t u) \leq \norm{(D_t^2 + L) u}^2,
  \end{equation*}
  where we have applied $L \geq 0$ and $L D_t = D_t L$ on $\mathcal{C}$.
  Therefore we finally obtain
  \begin{align*}
    \norm{\tilde\KG u}
    & \leq \norm{(-D_t^2 + L) u} + \norm{(2 V D_t + V^2) u}
      \leq \norm{(D_t^2 + L) u} + a \norm{N_0 u} + b \norm{u} \\
    & \leq (a + 1) \norm{N_0 u} + b \norm{u},
  \end{align*}
  using again the boundedness of $V$.

  (ii):
  We have to show that $\pm\im [\tilde\KG, N] \leq c N$ as quadratic forms on $\mathcal{C}$.
  However, on $\mathcal{C}$ we have (in the sense of quadratic forms)
  \begin{equation*}
    [\tilde\KG, N] = [\tilde\KG, \tilde\KG + 2 D_t^2] = 2 [\tilde\KG, D_t^2] = 0,
  \end{equation*}
  and thus $c = 0$, because $\tilde\KG$ does not depend on time.
\end{proof}

\begin{remark}
  If $V = 0$, an even simpler proof is possible.
  In this case we can write
  \begin{equation}\label{eq:P_tens_prod}
    \tilde\KG = -D_t^2 \otimes \one + \one \otimes L,
  \end{equation}
  and the essential self-adjointness of~$\tilde\KG$ on $C^2_\comp(\RR) \aotimes C^2_\comp(\Sigma)$ follows from the essential self-adjointness of~$D_t^2$ and $L$ on~$C^2_\comp(\RR)$ and~$C^2_\comp(\Sigma)$ by Thm.~\ref{thm:tens_prod_op}.
  Since we obviously have the inclusions $C^2_\comp(\RR) \aotimes C^2_\comp(\Sigma) \subset C^2_\comp(M) \subset \Dom \tilde\KG^*$, $\tilde\KG$ is even essentially self-adjoint on $C^2_\comp(M)$.
  As before, essential self-adjointness of $\KG$ on $C^2_\comp(M)$ follows by Lem.~\ref{lem:rel_selfadj}.
\end{remark}

\section{Hamiltonian formalism}

It is a simple exercise to rewrite~\eqref{eq:kg2} into an equation that is only first order in time:
Set $u_1(t) = u(t)$ and $u_2(t) = -(D_t + V) u(t)$, then
\begin{equation}\label{eq:1st_ord_kg}
  (\partial_t + \im B) \begin{pmatrix} u_1(t) \\ u_2(t) \end{pmatrix} = 0,
\end{equation}
where we defined
\begin{equation}\label{eq:def_B}
  B \defn \begin{pmatrix} V & \one \\ L & V \end{pmatrix}.
\end{equation}
Sometimes we call $\partial_t + \im B$ the \emph{first order Klein--Gordon operator}.

Let us denote by $(\,\cdot\;|\;\cdot\,)$ the canonical inner product on $L^2(\Sigma) \oplus L^2(\Sigma)$.
Although we use the same notation for the inner product on $L^2(M)$, no confusion should arise.
We introduce the \emph{charge matrix}
\begin{equation*}
  Q \defn \begin{pmatrix} 0 & \one \\ \one & 0 \end{pmatrix}.
\end{equation*}
It facilitates the definition of a (sesquilinear) charge form $(\,\cdot\;|\,Q\;\cdot\,)$ on $L^2(\Sigma) \oplus L^2(\Sigma)$.
The charge form plays essentially the role of the symplectic form in our complex setting.
The complex formalism is perhaps less known, however it is more convenient.
In particular, it is used by Gérard and Wrochna, \eg\ in~\cite{gerard-wrochna:hadamard}.

More importantly, we use $Q$ to define the \emph{classical Hamiltonian}
\begin{equation}\label{eq:H-QB}
  H \defn Q B = \begin{pmatrix} L & V \\ V & \one \end{pmatrix}
\end{equation}
with domain $(\Dom L) \oplus L^2(\Sigma)$.
\begin{proposition}\label{prop:H_selfadj}
  $H$ is self-adjoint in the sense of $L^2(\Sigma) \oplus L^2(\Sigma)$.
\end{proposition}
\begin{proof}
  $\begin{psmallmatrix} L & 0 \\ 0 & \one \end{psmallmatrix}$ is obviously self-adjoint, and $\begin{psmallmatrix} 0 & V \\ V & 0 \end{psmallmatrix}$ is self-adjoint and bounded.
\end{proof}

Physically realistic classical Hamiltonians should be positive, yet this cannot be guaranteed for $H$ as defined above.
Positivity can be spoiled if the electric potential $V$ is too large and it is easy to see that $H \geq 0$ if $L-V^2 \geq 0$.
A more precise result is the following:

\begin{proposition}\label{prop:H_pos}
  Let $C < 1$.
  $H \geq C$ if and only if $L - C - (1 - C)^{-1} V^2 \geq 0$ or, equivalently, $L - V^2 \geq C (1 - C)^{-1} V^2$.
  The implications continue to hold if replace all occurrences of~$\geq$ by~$>$.
\end{proposition}
\begin{proof}
  Decompose $H - C$ as
  \begin{equation}\label{eq:H_decomp}
    H - C =
    \begin{pmatrix} \one & (1 - C)^{-1} V \\ 0 & \one \end{pmatrix}
    \begin{pmatrix} L - C - (1 - C)^{-1} V^2 & 0 \\ 0 & 1 - C \end{pmatrix}
    \begin{pmatrix} \one & 0 \\ (1 - C)^{-1} V & \one \end{pmatrix}
  \end{equation}
  and note that the matrices on the left and right are invertible.
  The result follows immediately.
\end{proof}

Henceforth we will require:
\begin{assumption}\label{asm:H_pos}
  $H > 0$ or, equivalently, $L > V^2$.
\end{assumption}

We remark that this assumption can rule out the case $Y = 0$ on spacetimes with compact Cauchy surfaces $\Sigma$.

Since $H > 0$, we can consider the form domain of $H$ endowed with the scalar product given by $H$, the \emph{energy product}
\begin{equation*}
  (u \,|\, v)_\en \defn (u \,|\, H v),
\end{equation*}
as a Hilbert space in its own right.
We denote this space by $\mathcal{H}_\en$ and call it the \emph{energy space}.

\begin{proposition}
  $\mathcal{H}_\en = (\Dom L^\frac12) \oplus L^2(\Sigma)$.
\end{proposition}
\begin{proof}
  $\Dom H = (\Dom L) \oplus L^2(\Sigma)$ implies
  $\Dom H^\theta = (\Dom L^\theta) \oplus L^2(\Sigma)$ for $0 \leq \theta \leq 1$.
  Hence $\mathcal{H}_\en = \Dom H^\frac12 = (\Dom L^\frac12) \oplus L^2(\Sigma)$.
\end{proof}

\begin{remark}
  The original Hilbert space $L^2(\Sigma) \oplus L^2(\Sigma)$ plays a secondary role.
  The central role is played by $\mathcal{H}_\en$ and the scale of Hilbert spaces
  \begin{equation*}
    \mathcal{H}_\alpha \defn \abs{B}^{(1-\alpha)/2}_\en \mathcal{H}_\en,
    \quad
    \alpha \in \RR,
  \end{equation*}
  with scalar products
  \begin{equation*}
    (u\,|\,v)_\alpha \defn \bigl(u\,\big|\abs{B}_\en^{\alpha-1} v\bigr)_\en,
    \quad
    u, v \in \mathcal{H}_\alpha.
  \end{equation*}
  Of particular interest is the so-called dynamical space $\mathcal{H}_{\mathrm{dyn}} \defn \mathcal{H}_0$, see \eg~\cite{derezinski}.
\end{remark}

\begin{remark}
  $Q$ is not a bounded operator on $\mathcal{H}_\en$.
  However, it is easy to see that $Q$ can be defined with domain $(\Dom L^\frac12) \oplus (\Dom L^\frac12)$ and is closed on $\mathcal{H}_\en$.
\end{remark}

Consider $B$, given by~\eqref{eq:def_B}, an operator on $C^\infty_\comp(\Sigma) \oplus C^\infty_\comp(\Sigma)$.

\begin{proposition}\label{prop:B_selfadj}
  $B$ is essentially self-adjoint on $C^\infty_\comp(\Sigma) \oplus C^\infty_\comp(\Sigma)$ in the sense of $\mathcal{H}_\en$; its resolvent set is given by
  \begin{equation}\label{eq:reso_B}
    \reso(B) = \bigl\{ z \in \CC \;\big|\; \bigl( L - (V - z)^2 \bigr) (1+L)^{-\frac12} \;\text{is boundedly invertible} \bigr\}.
  \end{equation}
  We identify $B$ with its closure in~$\mathcal{H}_\en$.
\end{proposition}
\begin{proof}
  We have that $B$ is Hermitian in the sense of $\mathcal{H}_\en$ because
  \begin{equation*}
    (B u \,|\, v)_\en
    = (B u \,|\, H v)
    = (Q H u \,|\, H v)
    = (H u \,|\, Q H v)
    = (u \,|\, B v)_\en
  \end{equation*}
  for all $u, v \in C^\infty_\comp(\Sigma) \oplus C^\infty_\comp(\Sigma)$.
  Moreover, $B$ is closable, because $C^\infty_\comp(\Sigma)$ is a core for $L$.
  Its resolvent can be written as
  \begin{equation}\label{eq:B_decomp}
    (B - z)^{-1} =
    \begin{pmatrix} \one & 0 \\ z - V & \one \end{pmatrix}
    \begin{pmatrix*} 0 & \big( L - (V - z)^2 \big)^{-1} \\ \one & 0 \end{pmatrix*}
    \begin{pmatrix} \one & 0 \\ z - V & \one \end{pmatrix},
  \end{equation}
  which should be understood on the space $\mathcal{H}_\en$.
  Introduce
  \begin{equation*}
    U \defn \begin{pmatrix} (1+L)^{-\frac12} & 0 \\ 0 & \one \end{pmatrix},
  \end{equation*}
  which can be treated as a unitary from $L^2(\Sigma) \oplus L^2(\Sigma)$ to $\mathcal{H}_\en = (\Dom L^\frac12) \oplus L^2(\Sigma)$.
  Let us transport $(B - z)^{-1}$ onto $L^2(\Sigma) \oplus L^2(\Sigma)$:
  \begin{align*}
    U^{-1}(B - z)^{-1}U
    &=
    \begin{pmatrix} \one & 0 \\ (z - V) (1+L)^{-\frac12} & \one \end{pmatrix}
    \begin{pmatrix*} 0 & (1+L)^{\frac12}\big( L - (V - z)^2 \big)^{-1} \\ (1+L)^{-\frac12} & 0 \end{pmatrix*}
    \\&\qquad\cdot
    \begin{pmatrix} \one & 0 \\ (z - V) (1+L)^{-\frac12} & \one \end{pmatrix}.
  \end{align*}
  Hence we see that the resolvent set of $B$ is given by~\eqref{eq:reso_B}.
  To see that $B$ is self-adjoint, we need to find $z \in \CC$ above and below the real line such that
  \begin{equation*}
    (1+L)^{\frac12}\big( L - (V - z)^2 \big)^{-1} = (1+L)^{\frac12}{(L - z^2)}^{-1} {\big( \one - (V^2 - 2 z V) {(L - z^2)}^{-1} \big)}^{-1}
  \end{equation*}
  is well defined on $L^2(\Sigma)$.
  But for $z=\im y$ with $\abs{y}$ large enough
  \begin{equation*}
    \norm[\big]{(V^2 - 2 z V) {(L - z^2)}^{-1}} \leq \norm[\big]{V^2 - 2 z V} \norm[\big]{{(L - z^2)}^{-1}} < 1.
  \end{equation*}
  Hence we can use a Neumann series argument.
\end{proof}

\section{Inverses and bisolutions}

The concepts of an inverse or bisolution of $\partial_t + \im B$ or $\KG$ seem clear intuitively, but it is not obvious which functional spaces to choose in their definition, especially since we want to include low regularity situations.
To avoid such issues we will occasionally interpret the first order Klein--Gordon operator $\partial_t + \im B$ in the distributional sense, as a map from $\mathcal{D}'(M) \oplus \mathcal{D}'(M)$ into itself or as a map from $\mathcal{D}_\comp'(M) \oplus \mathcal{D}_\comp'(M)$ into itself.
Similarly, we will occasionally interpret the Klein--Gordon operator $\KG$ as a map from $\mathcal{D}'(M)$ into itself or as a map from $\mathcal{D}_\comp'(M)$ into itself.

Here, we will call an operator $E^\bullet$ from $C_\comp(M) \oplus C_\comp(M)$ to $\mathcal{D}'(M) \oplus \mathcal{D}'(M)$ an inverse, resp. a bisolution of $\partial_t + \im B$ if for $h \in C^2_\comp(M) \oplus C^2_\comp(M)$ we have
\begin{equation}
  (\partial_t + \im B) E^\bullet h = E^\bullet (\partial_t + \im B) h = h,\label{def1}
  \quad\text{resp.}\quad
  (\partial_t + \im B) E^\bullet h = E^\bullet (\partial_t + \im B) h = 0.
\end{equation}
(Note that $(\partial_t + \im B) h \in C_\comp(M) \oplus C_\comp(M)$, hence $E^\bullet (\partial_t + \im B) h$ makes sense in~\eqref{def1}. Besides, $\partial_t + \im B$ acting on $E^\bullet h$ can be understood in the distributional sense.)

An operator $G^\bullet$ from $C_\comp(M)$ to $\mathcal{D}'(M)$
will be called an inverse, resp. a bisolution of $\KG$ if for $f \in C^\infty_\comp(M)$ we have
\begin{equation}
  \KG G^\bullet f = G^\bullet \KG f = f,\label{def3}
  \quad\text{resp.}\quad
  \KG G^\bullet f = G^\bullet \KG f = 0.
\end{equation}
($\KG f \in C_\comp(M)$, hence $G^\bullet \KG f$ makes sense in~\eqref{def3}.
Besides, $\KG$ acting on $G^\bullet f$ can be understood in the distributional sense.)

Ultimately we are interested in propagators of the Klein--Gordon operator $\KG$, but the propagators of $\partial_t + \im B$ are closely related to those of~$\KG$.
Let us denote by $\pi_2$ the projection onto the second component:
\begin{equation}\label{eq:proj_maps}
  \pi_2 \begin{pmatrix} u_1 \\ u_2 \end{pmatrix} \defn u_2,
\end{equation}
We also define the embeddings
\begin{equation}\label{eq:emb_maps}
  \iota_2 u \defn \begin{pmatrix} 0 \\ u \end{pmatrix},
  \quad
  \rho u \defn \begin{pmatrix} u \\ -(D_t + V) u \end{pmatrix}.
\end{equation}
The maps $\pi_2, \rho, \iota_2$ can be understood between various spaces which should be inferred from the context.
A simple calculation shows that
\begin{equation*}
  \tilde\KG = \im \pi_2 (\partial_t + \im B) \rho
  \quad\text{and}\quad
  \KG = \im \beta^{-\frac12} \pi_2 (\partial_t + \im B) \rho \beta^{-\frac12}.
\end{equation*}
Consequently we find

\begin{proposition}\label{prop:E-vs-G}
  Suppose that $E^\bullet$ is either an inverse or a bisolution of $\partial_t + \im B$ in the sense of~\eqref{def1}.
  Then
  \begin{equation}\label{eq:G-E}
    G^\bullet = -\im \beta^\frac12 \pi_2 Q E^\bullet \iota_2 \beta^\frac12
  \end{equation}
  is an inverse resp. a bisolution of $\KG$ in the sense of~\eqref{def3}.
\end{proposition}
\begin{proof}
  Clearly, we have $\pi_2 Q \rho = \one$ and $\pi_2 \iota_2 = \one$.
  Since $E^\bullet$ is an inverse or bisolution, it satisfies
  \begin{align*}
    0
    &= \pi_2 Q (\partial_t + \im B) E^\bullet \iota_2 f
     = \pi_2 Q (\partial_t + \im B) \begin{pmatrix} u_1 \\ u_2 \end{pmatrix}\\
    &= \big( (\partial_t + \im V) u_1 + \im u_2 \big),
    \quad\text{where}\quad
    E^\bullet \iota_2 f = \begin{pmatrix} u_1 \\ u_2 \end{pmatrix},\,
    f \in C^\infty_\comp(M),
  \end{align*}
  \ie, $u_2 = -(D_t + V) u_1$.
  Applying $\rho \pi_2 Q$ to $(u_1, u_2)$, we find
  \begin{equation*}
    \rho \pi_2 Q \begin{pmatrix} u_1 \\ u_2 \end{pmatrix}
    = \rho u_1
    = \begin{pmatrix} u_1 \\ -(D_t + V) u_1 \end{pmatrix}
    = \begin{pmatrix} u_1 \\ u_2 \end{pmatrix},
  \end{equation*}
  and thus $\rho \pi_2 Q = \one$ on the range of $E^\bullet \iota_2$.
  Moreover,
  \begin{equation*}
    -\im (\partial_t + \im B) \rho u = \begin{pmatrix} 0 \\ \tilde\KG u \end{pmatrix},
  \end{equation*}
  \ie, the first component vanishes, and thus $\iota_2 \pi_2 = \one$ on the range of $(\partial_t + \im B) \rho$.
  Therefore, if $E^\bullet$ is an inverse, we find on $C^2_\comp(M)$
  \begin{align*}
    G^\bullet \KG &= \beta^\frac12 \pi_2 Q E^\bullet \iota_2 \pi_2 (\partial_t + \im B) \rho = \beta^\frac12 \pi_2 Q E^\bullet (\partial_t + \im B) \rho \beta^{-\frac12} = \beta^\frac12 \pi_2 Q \rho \beta^{-\frac12} = \one, \\
    \KG G^\bullet &= \beta^{-\frac12} \pi_2 (\partial_t + \im B) \rho \pi_2 Q E^\bullet \iota_2 \beta^\frac12 = \beta^{-\frac12} \pi_2 (\partial_t + \im B) E^\bullet \iota_2 \beta^\frac12 = \beta^{-\frac12} \pi_2 \iota_2 \beta^\frac12 = \one.
  \end{align*}
  It follows that $G^\bullet$ is an inverse.

  A similar calculation shows that $G^\bullet \KG = 0$ and $\KG G^\bullet = 0$ if $E^\bullet$ is a bisolution.
\end{proof}

\section{Classical propagators}

The most obvious examples of inverses and of a bisolution are furnished by the classical propagators for~\eqref{eq:1st_ord_kg}: the Pauli--Jordan propagator $E^\PJ$, the forward/retarded propagator $E^\vee$ and the backward/advanced propagator $E^\wedge$.
They are defined by the integral kernels
\begin{subequations}\label{eq:classical_kernels}
  \begin{align}
    E^\PJ(t - s) &\defn \e^{-\im (t - s) B}, \\
    E^\vee(t - s) &\defn \theta(t - s)\, \e^{-\im (t - s) B}, \\
    E^\wedge(t - s) &\defn -\theta(s - t)\, \e^{-\im (t - s) B}.
  \end{align}
\end{subequations}

Since $t \mapsto \e^{-\im t B} : \mathcal{H}_\en \to \mathcal{H}_\en$ are bounded, strongly continuously differentiable on the domain of $B$, it follows that

\begin{proposition}
  The operators $E^\PJ, E^\retadv$ defined by
  \begin{equation}\label{eq:E_op-E_kernel}
    (E^\bullet f)(t) = \int_\RR E^\bullet(t - s) f(s)\, \dif s,
    \quad
    f \in L^1(\RR; \mathcal{H}_\en),
  \end{equation}
  are bounded from $L^1(\RR; \mathcal{H}_\en)$ to $C_{\mathrm b}(\RR; \mathcal{H}_\en)$.
  $E^\retadv$ are inverses of $\partial_t + \im B$ and $E^\PJ$ is a bisolution of $\partial_t + \im B$.
\end{proposition}
Note that the relation $E^\PJ = E^\vee - E^\wedge$ holds.

Instead of the Banach space setting of the previous two proposition one might prefer to use a Hilbertian setting.
Define the `Japanese bracket' $\jnorm{t} \defn (1 + \abs{t}^2)^{1/2}$ and let $\mathcal{X}$ be a Hilbert space.
For $s \in \RR$, we consider the weighted spaces
\begin{equation*}
  \jnorm{t}^s L^2(\RR; \mathcal{X}).
\end{equation*}
For $s > 0$, we have the following rigging of the Hilbert space $L^2(\RR; \mathcal{X})$:
\begin{equation*}
  \jnorm{t}^{-s} L^2(\RR; \mathcal{X}) \subset L^2(\RR; \mathcal{X}) \subset \jnorm{t}^s L^2(\RR; \mathcal{X}).
\end{equation*}
Note that, for $s > \frac12$, we have the embeddings
\begin{equation*}
  \jnorm{t}^{-s} L^2(\RR; \mathcal{X}) \subset L^1(\RR; \mathcal{X})
  \quad\text{and}\quad
  \jnorm{t}^s L^2(\RR; \mathcal{X}) \supset C_{\mathrm b}(\RR; \mathcal{X}).
\end{equation*}
Therefore we can reinterpret the meaning of the classical propagators as follows:

\begin{proposition}
  For $s > \frac12$, the propagators $E^\PJ$, $E^\retadv$ are bounded operators from $\jnorm{t}^{-s} L^2(\RR; \mathcal{H}_\en)$ to $\jnorm{t}^s L^2(\RR; \mathcal{H}_\en)$.
\end{proposition}

We immediately use Prop.~\ref{prop:E-vs-G} to define the Pauli--Jordon propagator $G^\PJ$, the retarded propagator $G^\vee$ and the advanced propagator $G^\wedge$ of $\KG$ associated to the propagators $E^\PJ$, $E^\retadv$ of $\partial_t + \im B$.

\begin{proposition}
  For $s>\frac12$, the propagators $G^\PJ$, $G^\retadv$ are bounded operators from $\jnorm{t}^{-s} L^2(M)$ to $\jnorm{t}^s L^2(M)$.
  $G^\retadv$ are inverses of $\KG$ and $G^\PJ$ is a bisolution of $\KG$.
\end{proposition}
As for the classical propagators of $\partial_t + \im B$, we have the relation $G^\PJ = G^\vee - G^\wedge$.

\section{Non-classical propagators}

\begin{proposition}\label{prop:ker_B}
  $B$ has a trivial kernel on $\Dom B \subset \mathcal{H}_\en$.
\end{proposition}
\begin{proof}
  By assumption~\ref{asm:H_pos}, $H$ has a trivial kernel and clearly the same is true for $Q$.
  Now, $B = Q H$, see Eq.~\eqref{eq:H-QB}, so also $B$ has a trivial kernel.
\end{proof}

Using the spectral calculus on $\mathcal{H}_\en$, we can define complementary projectors $\Pi^{(\pm)}$ onto the positive and negative part of the spectrum of $B$.
These projections split the energy space as
\begin{equation*}
  \mathcal{H}_\en = \mathcal{H}_\en^{(+)} \oplus \mathcal{H}_\en^{(-)}.
\end{equation*}

The projectors $\Pi^{(\pm)}$ facilitate the definition of the non-classical propagators for~\eqref{eq:1st_ord_kg}: the positive and negative frequency bisolution/two-point function $E^{(\pm)}$, the Feynman propagator $E^\Feyn$ and the anti-Feynman propagator $E^\aFeyn$.
They are defined via their integral kernels as
\begin{subequations}\label{eq:non-classical_kernels}
  \begin{align*}
    E^{(\pm)}(t - s) &\defn \pm \e^{-\im (t - s) B} \Pi^{(\pm)}, \\
    E^\Feyn(t - s) &\defn \theta(t - s)\, \e^{-\im (t - s) B} \Pi^{(+)} - \theta(s - t)\, \e^{-\im (t - s) B} \Pi^{(-)}, \\
    E^\aFeyn(t - s) &\defn \theta(t - s)\, \e^{-\im (t - s) B} \Pi^{(-)} - \theta(s - t)\, \e^{-\im (t - s) B} \Pi^{(+)}.
  \end{align*}
\end{subequations}
As for the classical propagators, we can now deduce that

\begin{proposition}
  For $s > \frac12$, $E^{(\pm)}$ and $E^\FeynFeyn$ defined by the their kernels~\eqref{eq:non-classical_kernels} via Eq.~\eqref{eq:E_op-E_kernel} exist as bounded operators from $\jnorm{t}^{-s} L^2(\RR; \mathcal{H}_\en)$ to $\jnorm{t}^s L^2(\RR; \mathcal{H}_\en)$.
  $E^{(\pm)}$ are bisolutions and $E^\FeynFeyn$ are inverses of $\partial_t + \im B$.
\end{proposition}
We have the usual relations between the classical and non-classical propagators:
\begin{alignat*}{3}
  E^\Feyn &= E^\wedge \mkern-1mu + E^{(+)} = E^\vee \mkern-1mu + E^{(-)},
  &\qquad
  E^\Feyn + E^\aFeyn &= E^\vee \mkern-1mu + E^\wedge,
  &\qquad
  E^{(+)} \mkern-1mu - E^{(-)} = E^\PJ, \\
  E^\aFeyn &= E^\vee \mkern-1mu - E^{(+)} = E^\wedge \mkern-1mu - E^{(-)},
  &\qquad
  E^\Feyn - E^\aFeyn &= E^{(+)} \mkern-1mu + E^{(-)}.
\end{alignat*}

The corresponding propagators of $\KG$ have the following properties:
\begin{proposition}
  $G^\FeynFeyn$ induced via Eq.~\eqref{eq:G-E} and $G^{(\pm)} \defn \beta^\frac12 \pi_2 Q E^{(\pm)} \iota_2 \beta^\frac12$, are bounded operators from $\jnorm{t}^{-s} L^2(M)$ to $\jnorm{t}^s L^2(M)$.
  $G^{(\pm)}$ are bisolutions and $G^\FeynFeyn$ are inverses of $\KG$.
\end{proposition}
As for the propagators of $\partial_t + \im B$, we find for the propagators of $\KG$:
\begin{alignat*}{3}
  G^\Feyn &= G^\wedge \mkern-1mu + \im G^{(+)} = G^\vee \mkern-1mu + \im G^{(-)},
  &\qquad
  G^\Feyn + G^\aFeyn &= G^\vee \mkern-1mu + G^\wedge,
  &\qquad
  G^{(+)} \mkern-1mu - G^{(-)} = -\im G^\PJ, \\
  G^\aFeyn &= G^\vee \mkern-1mu - \im G^{(+)} = G^\wedge \mkern-1mu - \im G^{(-)},
  &\qquad
  G^\Feyn - G^\aFeyn &= \im G^{(+)} \mkern-1mu + \im G^{(-)}.
\end{alignat*}

Note that $\Pi^{(\pm)}$ are positive resp. negative with the respect to the charge form:

\begin{proposition}\label{prop:Q_pos}
  $\pm (u \,|\, Q \Pi^{(\pm)} u) \geq 0$ for all $u \in \mathcal{H}_\en$.
\end{proposition}
\begin{proof}
  Suppose $u = B v$ with $v \in \Dom B$.
  Then we can write
  \begin{equation}\label{eq:Q_pos}
    \mathord{\pm} (u \,|\, Q \Pi^{(\pm)} u)
    = \pm (H v \,|\, \Pi^{(\pm)} B\, v)
    = \pm (v \,|\, \Pi^{(\pm)} B\, v)_\en
    \geq 0,
  \end{equation}
  which is positive because the numerical range of $\Pi^{(\pm)} B$ is contained in the convex hull of its spectrum.
  Since $B$ has a trivial kernel, its range is dense in $\mathcal{H}_\en$ and we can extend~\eqref{eq:Q_pos} to the whole energy space (where~\eqref{eq:Q_pos} can be $+\infty$).
\end{proof}

It follows easily that
\begin{equation*}
  \iint \bigl( h(t) \,\big|\, Q E^{(\pm)}(t,s) h(s) \bigr)\, \dif s\, \dif t \geq 0
\end{equation*}
for $h \in \jnorm{t}^{-s} L^2(\mathcal{H}_\en)$ with $s > \frac12$.
This implies that the associated positive and negative frequency bisolution are positive:
\begin{proposition}
  We have
  \begin{equation*}
    (f \,|\, G^{(\pm)} f) = \int_M \conj{f}\, \bigl(G^{(\pm)} f\bigr)\, \abs{g}^\frac12 \geq 0
  \end{equation*}
  for $f \in \jnorm{t}^{-s} L^2(M)$ and $s > \frac12$.
\end{proposition}
\begin{proof}
  Using the relation
  \begin{align*}
    \int_M \conj{f}\, \bigl(G^{(\pm)} f\bigr)\, \abs{g}^\frac12
    &= \iint \bigl( f(t) \,\big|\, Q \pi_1 E^{(\pm)}(t,s) \iota_2 f(s) \bigr)\, \dif s\, \dif t \\
    &= \iint \bigl( \iota_2 f(t) \,\big|\, Q E^{(\pm)}(t,s) \iota_2 f(s) \bigr)\, \dif s\, \dif t,
  \end{align*}
  the desired result is immediate.
\end{proof}

\section{Limiting absorption principle}
\label{sec:limiting-absorption}

We define for all $z \in \im \RR$
\begin{equation*}
  B_z \defn B - z Z,
  \quad\text{where}\quad
  Z \defn \begin{pmatrix} 0 & 0 \\ 1 & 0 \end{pmatrix}.
\end{equation*}
Note that $Z$ is bounded on $\mathcal{H}_\en$.
By a simple modification of~\eqref{eq:reso_B}, we find that
\begin{equation}\label{eq:reso_B_z}
  \reso(B_z) = \bigl\{ \zeta \in \CC \;\big|\; \bigl( L - z - (V - \zeta)^2 \bigr) (1+L)^{-\frac12} \;\text{is boundedly invertible} \bigr\}.
\end{equation}

\begin{proposition}\label{prop:gap_B_z}
  Suppose that $L - V^2 \geq C > 0$.
  Then there exists $\alpha > 0$ such that the strip $\{ \zeta \in \CC \mid -\alpha \leq \Re\zeta \leq \alpha \}$ is contained in $\reso(B_z)$.
\end{proposition}
\begin{proof}
  We see from~\eqref{eq:reso_B_z} that a sufficient condition for $\zeta \in \reso(B_z)$ is that the real part of the numerical range of $L - z - (V - \zeta)^2$ is bounded away from zero, \viz, $L - \Re (V - \zeta)^2 \geq c' > 0$.
  This holds in particular if $L - (V - \Re\zeta)^2 \geq c > 0$ for some $c > 0$.
  Let use choose $c < C$ and set $\lambda = \Re\zeta$.
  The assumption $L - V^2 \geq C$ of the proposition implies that $L - (V - \lambda)^2 \geq C + 2 V \lambda - \lambda^2$.
  It is now not difficult to see that $C + 2 V \lambda - \lambda^2 \geq c > 0$ if $-\alpha \leq \lambda \leq \alpha$ with
  \begin{equation*}
    \alpha = -\norm{V} + \sqrt{C - c + \norm{V}^2}.\qedhere
  \end{equation*}
\end{proof}

It follows from Prop.~\ref{prop:H_pos} that $L - V^2 \geq C > 0$ implies $H \geq C' > 0$ and vice versa.
From now on we assume a strengthened version of Assumption~\ref{asm:H_pos}:

\begin{assumption}
  The classical Hamiltonian is bounded away from zero: $H \geq C > 0$.
\end{assumption}

We would like now to define spectral projections of $B_z$, generalizing $\Pi^{(\pm)}$, which were spectral projections of $B$.
This is somewhat more difficult, because $B_z$ are not self-adjoint, and hence we cannot use the standard spectral theorem, and $B_z$ are not bounded, hence we cannot directly use the standard holomorphic functional calculus.
However, the operators $B_z$ have good enough properties sufficient for a definiton of such projections.
In fact they are so-called \emph{bisectorial} operators and one can use results of \eg~\cite{winklmeier-wyss}, see also~\cite[Thm.~3.1]{bart-gohbarg-kaashoek}.
For the convenience of the reader, we sketch the construction of these projections in the next proposition:

\begin{proposition}
  The operators
  \begin{equation}\label{eq:spec_proj}
    \Pi_z^{(\pm)} \defn \slim_{\tau \to \infty} \frac{1}{2} \bigg( \one \pm \frac{1}{\uppi\im} \int_{-\im\tau}^{\im\tau} (B_z - \zeta)^{-1}\, \dif\zeta \bigg)
  \end{equation}
  are a pair of projections satifying
  \begin{equation*}
    \Pi^{(+)}_z + \Pi^{(-)}_z = \one
  \end{equation*}
  (\ie, they are complementary) and commuting with $B_z$.
  Moreover, they project onto the part of the spectrum in the left and right complex half-plane:
  \begin{equation}\label{pqo3}
    \spec(B_z \Pi_z^{\pm}) = \spec B_z\cap \{ z \in \CC \mid \pm\Re z \geq 0 \}.
  \end{equation}
\end{proposition}
\begin{proof}
  First we show that~\eqref{eq:spec_proj} are well defined.
  Using the resolvent identity and the functional calculus for self-adjoint operators, we get
  \begin{align*}
    \Pi_z^{(\pm)}
    &= \slim_{\tau \to \infty} \frac{1}{2} \bigg( \one \pm \frac{1}{\uppi\im} \int_{-\im\tau}^{\im\tau} (B - \zeta)^{-1} \big( \one + z Z (B_z - \zeta)^{-1} \big)\, \dif\zeta \bigg) \\
    &= \Pi^{(\pm)} \pm \frac{1}{2\uppi\im} \int_{-\im\infty}^{\im\infty} (B - \zeta)^{-1} z Z (B_z - \zeta)^{-1}\, \dif\zeta.
  \end{align*}
  The last integral converges absolutely because $Z$ is bounded and
  \begin{equation*}
    \norm*{(B - \im\lambda)^{-1}}_\en \leq \frac{c}{1 + \abs\lambda},
    \quad
    \norm*{(B_z - \im\lambda)^{-1}}_\en \leq \frac{c'}{1 + \abs\lambda},
    \quad
    \lambda \in \RR,
  \end{equation*}
  which follows from Prop.~\ref{prop:gap_B_z}.
  It follows that $\Pi_z^{(\pm)}$ are bounded operators and it is easy to see that they commute with $B_z$.

  Next we show that $\Pi^{(\pm)}_z$ are projections on a dense domain and hence everywhere.
  Let $\alpha$ as in Prop.~\ref{prop:gap_B_z} and choose $\beta, \beta'$ such that $0 < \beta < \beta' < \alpha$.
  It is a straightforward exercise to show that
  \begin{equation*}
    \Pi^{(+)} u = \frac{1}{2\uppi\im} \int_{\beta + \im\RR} \zeta^{-1} (B_z - \zeta)^{-1} B_z u\, \dif\zeta
  \end{equation*}
  for $u \in \Dom B_z$.
  Using the resolvent identity and Cauchy's theorem (as well as Fubini's theorem), we calculate for $u \in \Dom B_z^2$
  \begin{align*}
    \Pi^{(+)}_z{}^2 u
    &= \frac{1}{2\uppi\im} \int_{\beta + \im\RR} \zeta^{-1} (B_z - \zeta)^{-1}\, \bigg( \frac{1}{2\uppi\im} \int_{\beta' + \im\RR} \zeta'{}^{-1} (\zeta - \zeta')^{-1}\, \dif\zeta' \bigg) B_z^2 u\, \dif\zeta &\\&\quad - \frac{1}{2\uppi\im} \int_{\beta' + \im\RR} \zeta'{}^{-1} (B_z - \zeta')^{-1}\, \bigg( \frac{1}{2\uppi\im} \int_{\beta + \im\RR} \zeta^{-1} (\zeta - \zeta')^{-1}\, \dif\zeta \bigg) B_z^2 u\, \dif\zeta \\
    &= \frac{1}{2\uppi\im} \int_{\beta + \im\RR} \zeta^{-2} (B_z - \zeta)^{-1} B_z^2 u\, \dif\zeta.
  \end{align*}
  We can now apply the identity $\zeta^{-2} (B_z - \zeta)^{-1} B_z^2 = \zeta^{-1} (B_z - \zeta)^{-1} B_z + \zeta^{-2} B_z$ to find
  \begin{equation*}
    \frac{1}{2\uppi\im} \int_{\beta + \im\RR} \zeta^{-2} (B_z - \zeta)^{-1} B_z^2 u\, \dif\zeta = \Pi^{(+)}_z u + \frac{1}{2\uppi\im} \int_{\beta + \im\RR} \zeta^{-2} B_z u\, \dif\zeta,
  \end{equation*}
  where the last integral vanishes due to the residue theorem.
  It follows that $\Pi^{(+)}_z$ is a projection (and thus also $\Pi^{(-)}_z$) on $\Dom B_z^2$.
  Since $\Dom B_z^2$ is dense, $\Pi^{(\pm)}_z$ extend to bounded projections on~$\mathcal{H}_\en$.

  Finally we show that $\Pi^{(\pm)}_z$ have the claimed spectral properties~\eqref{pqo3}.
  For $\lambda \in \CC$, $-\alpha < \Re\lambda < \alpha$, we consider
  \begin{equation*}
    (B_z - \lambda)^{-1} \Pi^{(\pm)}_z = \frac{1}{2} (B_z - \lambda)^{-1} \pm \slim_{\tau \to \infty} \frac{1}{2\uppi\im} \int_{-\im\tau}^{\im\tau} (B_z - \lambda)^{-1} (B_z - \zeta)^{-1}\, \dif\zeta.
  \end{equation*}
  These extend as analytic functions with values in bounded operators for $\pm\Re\lambda < 0$:
  \begin{align*}
    \MoveEqLeft\slim_{\tau \to \infty} \frac{1}{2\uppi\im} \int_{-\im\tau}^{\im\tau} (B_z - \lambda)^{-1} (B_z - \zeta)^{-1}\, \dif\zeta \\
    &= \slim_{\tau \to \infty} \frac{1}{2\uppi\im} \int_{-\im\tau}^{\im\tau} (\zeta - \lambda)^{-1} \big( (B_z - \zeta)^{-1} - (B_z - \lambda)^{-1} \big)\, \dif\zeta \\
    &= \mp\frac{1}{2} (B_z - \lambda)^{-1} + \frac{1}{2\uppi\im} \int_{\im\RR} (\zeta - \lambda)^{-1} (B_z - \zeta)^{-1}\, \dif\zeta
  \end{align*}
  by the resolvent identity and the residue theorem.
\end{proof}

\begin{proposition}
  Suppose $L - 2 V^2 \geq 0$.
  Then $\pm B_z$ are maximally dissipative on $\Pi^{(\pm)}_z \mathcal{H}_\en$ for $\Im z \geq 0$, and $\pm B_z$ are maximally dissipative on $\Pi^{(\mp)}_z \mathcal{H}_\en$ for $\Im z \leq 0$.
\end{proposition}
\begin{proof}
  Let $\Im z \geq 0$; the proof of the other case is analogous.

  On $\mathcal{H}_\en$, $B$ is self-adjoint, whence maximally dissipative, and $Z$ is bounded (and thus it has $B$-bound~$0$).
  Suppose for a moment that $-z Z$ is dissipative on $\Pi^{(\pm)}_z \mathcal{H}_\en$.
  By a standard argument, see \eg~\cite[Thm. V-4.3]{kato}, we can then deduce that also $B - z Z$ is maximally dissipative on $\Pi^{(\pm)}_z \mathcal{H}_\en$.

  It remains to show that $-z Z$ is dissipative on $\Pi^{(\pm)}_z \mathcal{H}_\en$, \viz,
  \begin{equation*}
    0 \leq \pm \Im \left( \Pi^{(\pm)}_z u \,\middle|\, z Z \Pi^{(\pm)}_z u \right)_\en = \pm \Im z \Re \left( \Pi^{(\pm)}_z u \,\middle|\, Z \Pi^{(\pm)}_z u \right)_\en.
  \end{equation*}
  Given that $\Pi^{(\pm)}_z$ are complementary projections and using~\eqref{eq:spec_proj}, this is equivalent to
  \begin{align}
    0 &\leq
    \Pi^{(+)\,*}_z (H Z + Z^* H) \Pi^{(+)}_z - \Pi^{(-)\,*}_z (H Z + Z^* H) \Pi^{(-)}_z \label{eq:Z_ineq}\\
    &= \big(\Pi^{(+)\,*}_z - \Pi^{(-)\,*}_z\big) (H Z + Z^* H) + (H Z + Z^* H) \big(\Pi^{(+)}_z - \Pi^{(-)}_z\big) \notag\\
    &= \frac{1}{\uppi\im} \int_{\im\RR} \big( (B_z^* + \zeta)^{-1} (H Z + Z^* H) + (H Z + Z^* H)\, (B_z - \zeta)^{-1} \big)\, \dif \zeta \notag\\
    &= \frac{1}{\uppi\im} \int_{\im\RR} (B_z^* + \zeta)^{-1} \big( B_z^* (H Z + Z^* H) + (H Z + Z^* H) B_z \big)\, (B_z - \zeta)^{-1}\, \dif \zeta \notag\\
    &= \frac{1}{\uppi\im} \int_{\im\RR} (B_z^* + \zeta)^{-1} \big( B^* (H Z + Z^* H) + (H Z + Z^* H) B \big)\, (B_z - \zeta)^{-1}\, \dif \zeta. \label{eq:Z_ineq2}
  \end{align}
  We calculate
  \begin{equation*}
    B^* (H Z + Z^* H) + (H Z + Z^* H) B = 2 \begin{pmatrix} L + 2 V^2 & 2 V \\ 2 V & \one \end{pmatrix}.
  \end{equation*}
  Hence, for $L - 2 V^2 \geq 0$, the integrand in~\eqref{eq:Z_ineq2} is positive and we see that the inequality~\ref{eq:Z_ineq} holds.
\end{proof}

We wish to remark that the requirement $L - 2 V^2 \geq 0$ in the proposition is probably not optimal.
Nevertheless, for remainder of this section we assume:

\begin{assumption}
  $L - 2 V^2 \geq 0$.
\end{assumption}

Since maximally dissipative operators generate strongly continuous semigroups of contractions, we may thus define
\begin{equation*}
  E_z^\Feyn(t - s) \defn
  \begin{cases}
    \theta(t - s)\, \e^{-\im (t - s) B_z} \Pi^{(+)}_z - \theta(s - t)\, \e^{-\im (t - s) B_z} \Pi^{(-)}_z, &\quad \text{for $\Im z < 0$}, \\
    \theta(t - s)\, \e^{-\im (t - s) B_z} \Pi^{(-)}_z - \theta(s - t)\, \e^{-\im (t - s) B_z} \Pi^{(+)}_z, &\quad \text{for $\Im z > 0$}.
  \end{cases}
\end{equation*}
Note that $E_z^\Feyn(t - s)$ is the integral kernel of an inverse $E_z^\Feyn$ of $\partial_t + \im B - z Z$.
We denote by $G_z^\Feyn$ the corresponding inverse of $\KG - z$.

\begin{proposition}\label{prop:slim_E_z}
  We have
  \begin{equation*}
    E^\Feyn = \slim_{\varepsilon \searrow 0} E_{\im\varepsilon}^\Feyn,
  \end{equation*}
  in the sense of operators from $\jnorm{t}^{-s} L^2(\RR; \mathcal{H}_\en)$ to $\jnorm{t}^s L^2(\RR; \mathcal{H}_\en)$ for $s > \frac12$.
\end{proposition}
\begin{proof}
  Suppose that $t > 0$.
  Using the fundamental theorem of calculus, we find
  \begin{align*}
    \norm*{E_z^\Feyn(t) u - E^\Feyn(t) u}_\en
    &= \norm*{\int_0^t \frac{\dif}{\dif s} \big(E_z^\Feyn(t - s) E^\Feyn(s) \big)\, u\, \dif s}_\en \\
    &= \norm*{\int_0^t \big(E_z^\Feyn(t - s) (B_z - B) E^\Feyn(s) \big)\, u\, \dif s}_\en \\
    &= \norm*{z \int_0^t \big(E_z^\Feyn(t - s) Z E^\Feyn(s) \big)\, u\, \dif s}_\en \\
    &\leq \abs{t z} \norm{u}_\en
  \end{align*}
  for $u \in \Dom B$.
  The same bound can be found for $t < 0$.

  Since $\norm{E_z^\Feyn(t)} \leq 1$ and $\Dom B$ dense in $\mathcal{H}_\en$,
  \begin{equation*}
    E^\Feyn(t) = \slim_{\varepsilon \searrow 0} E_{\im\varepsilon}^\Feyn(t)
  \end{equation*}
  on $\mathcal{H}_\en$ uniformly for $t$ in bounded subsets of $(-\infty, 0)$ and $(0, \infty)$.
  In particular the convergence is pointwise, thus by Lebesgue's dominated convergence theorem
  \begin{equation*}
    \lim_{\varepsilon \searrow 0}\, \norm*{E_{\im\varepsilon}^\Feyn u - E^\Feyn u}_{C_{\mathrm b}(\RR; \mathcal{H}_\en)} = 0
  \end{equation*}
  for $u \in L^1(\RR; \mathcal{H}_\en)$.
  Using the embeddings
  \begin{equation*}
    \jnorm{t}^{-s} L^2(\RR; \mathcal{H}_\en) \subset L^1(\RR; \mathcal{H}_\en)
    \quad\text{and}\quad
    \jnorm{t}^s L^2(\RR; \mathcal{H}_\en) \supset C_{\mathrm b}(\RR; \mathcal{H}_\en)
  \end{equation*}
  for $s > \frac12$, we are done.
\end{proof}

Recall that $\KG$ is essentially selfadjoint on $C^2_\comp(M)$ in the sense of $L^2(M)$. Thus its closure $\KG^\cl$ has a real spectrum and for $\Im z \neq 0$ the resolvent $(\KG^\cl - z)^{-1}$ is well defined as a bounded operator on $L^2(M)$.

We have the following interpretation of the Feynman propagator of $\KG$:
\begin{theorem}\label{thm:main-theorem}
  We have
  \begin{equation*}
    G^\Feyn = \slim_{\varepsilon \searrow 0} (\KG^\cl - \im\varepsilon)^{-1},
  \end{equation*}
  \begin{equation*}
    G^\aFeyn = \slim_{\varepsilon \nearrow 0} (\KG^\cl - \im\varepsilon)^{-1}.
  \end{equation*}
  in the sense of operators from $\jnorm{t}^{-s} L^2(M)$ to $\jnorm{t}^s L^2(M)$ for $s > \frac12$.
\end{theorem}
\begin{proof}
  As a consequence of Prop.~\ref{prop:slim_E_z}, we have
  \begin{equation*}
    G^\Feyn = \slim_{\varepsilon \searrow 0} G_{\im\varepsilon}^\Feyn
  \end{equation*}
  It is now not difficult to see that
  \begin{equation*}
    G_z^\Feyn = (\KG^\cl - z)^{-1}
  \end{equation*}
  for $z \in \im\RR$.
\end{proof}

Using the language from the theory of Schrödinger operators, this means that the limiting absorption principle holds for $\KG$ at $0$ and that it yields the Feynman propagator.

\begin{remark}
  Before we continue, let us remark that if the electric potential $V$ vanishes one can derive the limiting absorption principle for $\KG$ by a simpler argument.
  Then one can use the tensor product structure~\eqref{eq:P_tens_prod} of $\tilde\KG$ to derive the limiting absorption principle for $\KG$ from the fact that
  \begin{equation*}
    (\partial_t^2 + \lambda \pm \im0)^{-1} \defn \slim_{\varepsilon \searrow 0}\, (\partial_t^2 + \lambda \pm \im\varepsilon)^{-1},
    \quad
    \lambda \in \RR \setminus \{0\},
  \end{equation*}
  exists as a bounded operator from $\jnorm{t}^{-s} L^2(\RR)$ to $\jnorm{t}^s L^2(\RR)$ for $s > \frac12$.
  See, for example, \cite[Chap.~5]{ben-artzi} for results on the limiting absorption principle for operators of the form $H = H_1 \otimes \one + \one \otimes H_2$.
\end{remark}

\section{Wick rotation}

Let $0 \leq \theta \leq \uppi$.
Suppose we replace the metric $g$ in~\eqref{eq:metric} by
\begin{equation*}
  g_\theta \defn -\e^{-2 \im \theta} \beta\, \dif t^2 + g_\Sigma
\end{equation*}
and the electric potential $V$ by $V_\theta \defn \e^{-\im \theta} V$.
This replacement is called \emph{Wick rotation}.
The value $\theta = \frac{\uppi}{2}$ corresponds to the Riemannian metric
\begin{equation*}
  g_{\uppi/2} = g_R = \beta\, \dif t^2 + g_\Sigma.
\end{equation*}

Constructing a Wick rotated version $B_\theta$ of $B$ as in~\eqref{eq:1st_ord_kg}, we define
\begin{equation*}
  B_\theta \defn \e^{-\im \theta} B.
\end{equation*}
For our purposes we could also take this equation as our definition of Wick rotation.

\begin{proposition}
  For $\theta \in [0, \uppi]$, $\pm B_\theta$ are maximally dissipative on $\mathcal{H}_\en^{(\pm)}$.
  In other words, $\pm B_\theta$ are generators of strongly continuous semigroups of contractions on $\mathcal{H}_\en^{(\pm)}$.
\end{proposition}
\begin{proof}
  We calculate
  \begin{equation*}
    \pm \Im \left( \Pi^{(\pm)} u \,\middle|\, B_\theta \Pi^{(\pm)} u \right)_\en
    = \mp \sin \theta \left( \Pi^{(\pm)} u \,\middle|\, B \Pi^{(\pm)} u \right)_\en
    \leq 0
  \end{equation*}
  for $\theta \in [0, \uppi]$ and thus $\pm B_\theta$ are dissipative.
  To see whether $\pm B_\theta$ are even maximally dissipative, we check that
  that the range of $\pm \e^{-\im \theta} B - \zeta$ is dense in $\mathcal{H}_\en^{(\pm)}$ for $\Im \zeta > 0$.
  Since the spectrum of $B$ restricted to $\mathcal{H}_\en^{(\pm)}$ does not include $\pm \e^{\im \theta} \zeta$, this is automatic.
\end{proof}

Therefore
\begin{equation*}
  \e^{-\im (t - s) B_\theta} \Pi^{(\pm)},
  \quad\text{for $\pm t \geq \pm s$},
\end{equation*}
are bounded (and even exponentially decaying) on $\mathcal{H}_\en$ and we may define a Wick rotated analog of the Feynman propagator:
\begin{equation*}
  E_\theta^\Feyn(t - s) \defn \theta(t - s)\, \e^{-\im (t - s) B_\theta} \Pi^{(+)} - \theta(s - t)\, \e^{-\im (t - s) B_\theta} \Pi^{(-)}.
\end{equation*}

Note that, as $\theta \!\searrow\! 0$,
the Wick rotated Feynman propagator converges strongly to the unrotated propagator:

\begin{proposition}
  We have
  \begin{equation*}
    E^\Feyn = \slim_{\theta \searrow 0} E_\theta^\Feyn,
  \end{equation*}
  in the sense of operators from $\jnorm{t}^{-s} L^2(\RR; \mathcal{H}_\en)$ to $\jnorm{t}^s L^2(\RR; \mathcal{H}_\en)$ for $s > \frac12$.
\end{proposition}
\begin{proof}
  This may be shown in a similar way as Prop.~\ref{prop:slim_E_z}.
\end{proof}

As a consequence we have the corresponding convergence for the Feynman propagator of $\KG$:

\begin{theorem}
  We have
  \begin{equation*}
    G^\Feyn = \slim_{\theta \searrow 0} G_\theta^\Feyn
  \end{equation*}
  in the sense of operators from $\jnorm{t}^{-s} L^2(M)$ to $\jnorm{t}^s L^2(M)$ for $s > \frac12$.
\end{theorem}

\begin{remark}
  Note that the Feynman propagator is distinguished by the fact that it can be Wick rotated.
  Wick rotated versions of the positive and negative frequency bisolutions $E^{(\pm)}$ (resp. $G^{(\pm)}$), for example, cannot be defined as bounded operators using the methods described above.
  The obstruction is that $\e^{-\im t B_\theta} \Pi^{(\pm)}$ are contractive \emph{semi}groups but not groups (\ie, we are restricted to $\pm t \geq 0$).
\end{remark}

\begin{acknowledgments}
  We would like to thank Christian Gérard, András Vasy and Michał\linebreak Wrochna for useful discussions.
  We also thank Bernard Kay for pointing out to us the work of Rumpf~\cite{rumpf1,rumpf2}.
  The work of D.S. was supported by a grant of the Polish National Science Center (NCN) based on the decision no. DEC-2015/16/S/ST1/00473.
  The work of J.D. was supported by the National Science Center under the grant UMO-2014/15/B/ST1/00126.
\end{acknowledgments}

\appendix
\section{A few theorems}

\begin{lemma}\label{lem:rel_selfadj}
  Let $\mathcal{H}$ be a Hilbert space and $\mathcal{D} \subset \mathcal{H}$ a dense subset.
  Suppose that $T : \mathcal{D} \to \mathcal{H}$ is essentially self-adjoint on $\mathcal{D}$, and $S : \mathcal{H} \to \mathcal{H}$ is bounded and boundedly invertible.
  Then $S^* T S$ is essentially self-adjoint on $S^{-1} \mathcal{D}$.
\end{lemma}

\begin{theorem}[see \eg~{\cite[Chap. VIII.10]{reed-1}}]\label{thm:tens_prod_op}
  Let $\mathcal{H}_1$, $\mathcal{H}_2$ be Hilbert spaces and $T_1$, $T_2$ densely defined operators on $\mathcal{H}_1$ and $\mathcal{H}_2$.
  Suppose that $T_1$ and $T_2$ are essentially self-adjoint on $\Dom T_1$ and $\Dom T_2$.
  Then $T = T_1 \otimes \one + \one \otimes T_2$ is essentially self-adjoint on the algebraic tensor product of the domains, $\Dom T_1 \aotimes \Dom T_2$.
\end{theorem}

\begin{theorem}[Nelson's Commutator Theorem, see \eg~\cite{faris}]\label{thm:nelson}
  Let $T$ be a Hermitian operator and $N \geq 0$ a positive self-adjoint operator.
  Let $\mathcal{C}$ be a core for $N$ such that $\mathcal{C} \subset \Dom T$.
  Assume that the following two estimates hold:
  \begin{enumerate}
    \item[\normalfont(i)] $\norm{T f} \leq a \norm{N f} + b \norm{f}$\; for $f \in \mathcal{C}$,
    \item[\normalfont(ii)] $\pm\im [T, N] \leq c N$\; as quadratic forms on $\mathcal{C}$.
  \end{enumerate}
  Then $T$ is essentially self-adjoint on $\mathcal{C}$.
\end{theorem}



\small
\begin{thebibliography}{10}

\bibitem{bgp}
B{\"a}r, C., Ginoux, N., Pf{\"a}ffle, F.:
  \href{http://dx.doi.org/10.4171/037}{\emph{Wave Equations on {Lorentzian}
  Manifolds and Quantization}}.
\newblock ESI Lectures in Mathematical Physics. European Mathematical Society,
  (2007).
\newblock \href{http://arxiv.org/abs/0806.1036}{{\ttfamily arXiv:0806.1036
  [math.DG]}}

\bibitem{bart-gohbarg-kaashoek}
Bart, H., Gohbarg, I., Kaashoek, M.: Wiener--{Hopf} Factorization, Inverse
  {Fourier} Transforms and Exponentially Dichotomous Operators.
\newblock \href{http://dx.doi.org/10.1016/0022-1236(86)90055-8}{Journal of
  Functional Analysis \textbf{68}, 1--42 (1986)}

\bibitem{ben-artzi}
Ben-Artzi, M.: \href{http://dx.doi.org/10.1007/978-3-0348-0024-2_3}{Smooth
  Spectral Calculus}.
\newblock In: \emph{Partial Differential Equations and Spectral Theory}, no.
  211 in Operator Theory: Advances and Applications, chap.~3, pp. 119--182.
  Birkh{\"a}user,  (2011)

\bibitem{bogoliubov}
Bogoliubov, N.N., Shirkov, D.V.: \emph{Introduction to the Theory of Quantized
  Fields}.
\newblock John Wiley \& Sons, 3 edn.,  (1980)

\bibitem{brum-fredenhagen}
Brum, M., Fredenhagen, K.: `Vacuum-like' {Hadamard} States for Quantum Fields
  on Curved Spacetimes.
\newblock \href{http://dx.doi.org/10.1088/0264-9381/31/2/025024}{Classical and
  Quantum Gravity \textbf{31}, 025024 (2014)}.
\newblock \href{http://arxiv.org/abs/1307.0482}{{\ttfamily arXiv:1307.0482
  [gr-qc]}}

\bibitem{chernoff}
Chernoff, P.R.: Essential Self-Adjointness of Powers of Generators of
  Hyperbolic Equations.
\newblock \href{http://dx.doi.org/10.1016/0022-1236(73)90003-7}{Journal of
  Functional Analysis \textbf{12}, 401--414 (1973)}

\bibitem{derezinski}
Derezi{\'n}ski, J., G{\'e}rard, C.: \emph{Mathematics of Quantization and
  Quantum Fields}.
\newblock Cambridge Monographs on Mathematical Physics. Cambridge University
  Press,  (2013)

\bibitem{derezinski-siemssen:evolution}
Derezi{\'n}ski, J., Siemssen, D.: An Evolution Equation Approach to the
  {Klein}--{Gordon} Operator on Curved Spacetime.
\newblock \href{http://arxiv.org/abs/1709.03911}{{\ttfamily arXiv:1709.03911
  [math-ph]}}

\bibitem{derezinski-siemssen:feynman}
Derezi{\'n}ski, J., Siemssen, D.: The Feynman Propagator on Curved Spacetimes
  (in preparation)

\bibitem{devinatz}
Devinatz, A.: Essential Self-Adjointness of {Schr{\"o}dinger}-Type Operators.
\newblock \href{http://dx.doi.org/10.1016/0022-1236(77)90032-5}{Journal of
  Functional Analysis \textbf{25}, 58--69 (1977)}

\bibitem{duistermaat}
Duistermaat, J.J., H{\"o}rmander, L.: {Fourier} Integral Operators. {II}.
\newblock \href{http://dx.doi.org/10.1007/BF02392165}{Acta Mathematica
  \textbf{128}, 183--269 (1972)}

\bibitem{faris}
Faris, W.G., Lavine, R.B.: Commutators and Self-Adjointness of {Hamiltonian}
  Operators.
\newblock \href{http://dx.doi.org/10.1007/BF01646453}{Communications in
  Mathematical Physics \textbf{35}, 39--48 (1974)}

\bibitem{fewster-verch}
Fewster, C.J., Verch, R.: On a Recent Construction of 'Vacuum-Like' Quantum
  Field States in Curved Spacetime.
\newblock \href{http://dx.doi.org/10.1088/0264-9381/29/20/205017}{Classical and
  Quantum Gravity \textbf{29}, 205017 (2012)}.
\newblock \href{http://arxiv.org/abs/1206.1562}{{\ttfamily arXiv:1206.1562
  [math-ph]}}

\bibitem{frehse}
Frehse, J.: Essential Selfadjointness of Singular Elliptic Operators.
\newblock \href{http://dx.doi.org/10.1007/BF02584723}{Boletim da Socieda de
  Brasileira de Matem{\'a}tica \textbf{8}, 87--107 (1977)}

\bibitem{fulling:hadamard}
Fulling, S.A., Narcowich, F.J., Wald, R.M.: Singularity Structure of the
  Two-Point Function in Quantum Field Theory in Curved Spacetime, {II}.
\newblock \href{http://dx.doi.org/10.1016/0003-4916(81)90098-1}{Annals of
  Physics \textbf{136}, 243--272 (1981)}

\bibitem{gell-redman}
Gell-Redman, J., Haber, N., Vasy, A.: The {Feynman} Propagator on Perturbations
  of {Minkowski} Space.
\newblock \href{http://dx.doi.org/10.1007/s00220-015-2520-8}{Communications in
  Mathematical Physics \textbf{342}, 333--384 (2016)}

\bibitem{gerard-wrochna:hadamard}
G{\'e}rard, C., Wrochna, M.: Construction of {Hadamard} States by
  Pseudo-Differential Calculus.
\newblock \href{http://dx.doi.org/10.1007/s00220-013-1824-9}{Communications in
  Mathematical Physics \textbf{325}, 713--755 (2014)}.
\newblock \href{http://arxiv.org/abs/1209.2604}{{\ttfamily arXiv:1209.2604
  [math-ph]}}

\bibitem{gerard-wrochna:feynman}
G{\'e}rard, C., Wrochna, M.: The Massive {Feynman} Propagator on Asymptotically
  {Minkowski} Spacetimes.
\newblock \href{http://arxiv.org/abs/1609.00192}{{\ttfamily arXiv:1609.00192
  [math-ph]}}

\bibitem{itzykson-zuber}
Itzykson, C., Zuber, J.B.: \emph{Quantum Field Theory}.
\newblock International Series in Pure and Applied Physics. McGraw-Hill,
  (1980)

\bibitem{kato}
Kato, T.: \emph{Perturbation Theory for Linear Operators}.
\newblock No. 132 in Grundlehren der mathematischen Wissenschaften. Springer,
  (1980)

\bibitem{kay}
Kay, B.S.: Linear Spin-Zero Quantum Fields in External Gravitational and Scalar
  Fields.
\newblock \href{http://dx.doi.org/10.1007/BF01940330}{Communications in
  Mathematical Physics \textbf{62}, 55--70 (1978)}

\bibitem{khavkine-moretti}
Khavkine, I., Moretti, V.: \emph{Algebraic QFT in Curved Spacetime and
  Quasifree Hadamard States: An Introduction},
  \href{http://dx.doi.org/10.1007/978-3-319-21353-8_5}{pp. 191--251}.
\newblock Springer,  (2015)

\bibitem{moretti:hadamard}
Moretti, V.: Quantum Out-States Holographically Induced by Asymptotic Flatness:
  Invariance under Spacetime Symmetries, Energy Positivity and {Hadamard}
  Property.
\newblock \href{http://dx.doi.org/10.1007/s00220-008-0415-7}{Communications in
  Mathematical Physics \textbf{279}, 31--75 (2008)}.
\newblock \href{http://arxiv.org/abs/gr-qc/0610143}{{\ttfamily
  arXiv:gr-qc/0610143}}

\bibitem{radzikowski}
Radzikowski, M.J.: Micro-Local Approach to the {Hadamard} Condition in Quantum
  Field Theory on Curved Space-Time.
\newblock \href{http://dx.doi.org/10.1007/BF02100096}{Communications in
  Mathematical Physics \textbf{179}, 529--553 (1996)}

\bibitem{reed-1}
Reed, M., Simon, B.: \emph{Methods of Modern Mathematical Physics {I}:
  Functional Analysis}.
\newblock Academic Press,  (1980)

\bibitem{rumpf2}
Rumpf, H.: Vacuum State and Particle Creation in External Electromagnetic
  Fields.
\newblock \href{http://dx.doi.org/10.5169/seals-115111}{Helvetica Physica Acta
  \textbf{53}, 85--111 (1980)}

\bibitem{rumpf1}
Rumpf, H., Urbantke, H.K.: Covariant ``In--Out'' Formalism for Creation by
  External Fields.
\newblock \href{http://dx.doi.org/10.1016/0003-4916(78)90273-7}{Annals of
  Physics \textbf{114}, 332--355 (1978)}

\bibitem{sahlmann-verch:passive}
Sahlmann, H., Verch, R.: Passivity and Microlocal Spectrum Condition.
\newblock \href{http://dx.doi.org/10.1007/s002200000297}{Communications in
  Mathematical Physics \textbf{214}, 705--731 (2000)}.
\newblock \href{http://arxiv.org/abs/math-ph/0002021}{{\ttfamily
  arXiv:math-ph/0002021}}

\bibitem{shubin}
Shubin, M.: Essential Self-Adjointness for Semi-Bounded Magnetic
  {Schr{\"o}dinger} Operators on Non-Compact Manifolds.
\newblock \href{http://dx.doi.org/10.1006/jfan.2001.3778}{Journal of Functional
  Analysis \textbf{186}, 92--116 (2001)}

\bibitem{strichartz}
Strichartz, R.S.: Analysis of the {Laplacian} on the Complete {Riemannian}
  Manifold.
\newblock \href{http://dx.doi.org/10.1016/0022-1236(83)90090-3}{Journal of
  Functional Analysis \textbf{52}, 48--79 (1983)}

\bibitem{vasy}
Vasy, A.: On the Positivity of Propagator Differences.
\newblock \href{http://arxiv.org/abs/1411.7242}{{\ttfamily arXiv:1411.7242
  [math.AP]}}

\bibitem{winklmeier-wyss}
Winklmeier, M., Wyss, C.: On the Spectral Decomposition of Dichotomous and
  Bisectorial Operators.
\newblock \href{http://dx.doi.org/10.1007/s00020-015-2218-5}{Integral Equations
  and Operator Theory \textbf{82}, 119--150 (2015)}

\end{thebibliography}
\end{document}